\documentclass[sigconf,natbib=false]{acmart}

\usepackage{url}   
\usepackage{breakurl}
\usepackage{longtable}
\usepackage{algorithm}
\usepackage{algpseudocode}
\usepackage{bm}

\AtBeginDocument{%
  }

\setcopyright{acmlicensed}
\copyrightyear{2024}
\acmYear{2024}
\acmDOI{XXXXXXX.XXXXXXX}

\RequirePackage[
  datamodel=acmdatamodel,
  style=acmnumeric,
  ]{biblatex}

\begin{document}

\title{Adaptive Flip Graph Algorithm for Matrix Multiplication}

\author{Yamato Arai}
\email{arai-yamato206@g.ecc.u-tokyo.ac.jp}
\affiliation{%
  \institution{Department of Integrated Sciences, 
  University of Tokyo}
  \streetaddress{P.O. Box 1212}
  \city{Tokyo}
  \state{Tokyo}
  \country{Japan}
  \postcode{43017-6221}
}

\author{Yuma Ichikawa}
\email{ichikawa-yuma1@g.ecc.u-tokyo.ac.jp}
\affiliation{
  \institution{Graduate School of Arts and Sciences, 
  University of Tokyo}
  \streetaddress{P.O. Box 1212}
  \city{Tokyo}
  \state{Tokyo}
  \country{Japan}
  \postcode{43017-6221}
}
\email{ichikawa.yuma@fujitsu.com}
\affiliation{
  \institution{Fujitsu Laboratories}
  \streetaddress{P.O. Box 1212}
  \city{Tokyo}
  \state{Tokyo}
  \country{Japan}
  \postcode{43017-6221}
}

\author{Koji Hukushima}
\email{k-hukushima@g.ecc.u-tokyo.ac.jp}
\affiliation{%
  \institution{Graduate School of Arts and Sciences,
  University of Tokyo}
  \streetaddress{P.O. Box 1212}
  \city{Tokyo}
  \state{Tokyo}
  \country{Japan}
  \postcode{43017-6221}
}

\renewcommand{\shortauthors}{Arai et al.}

\begin{abstract}
    This study proposes the ``adaptive flip graph algorithm'', which combines adaptive searches with the flip graph algorithm for finding fast and efficient methods for matrix multiplication.
    The adaptive flip graph algorithm addresses the inherent limitations of exploration and inefficient search encountered in the original flip graph algorithm, particularly when dealing with large matrix multiplication.   
    For the limitation of exploration, the proposed algorithm adaptively transitions over the flip graph, introducing a flexibility that does not strictly reduce the number of multiplications. 
    Concerning the issue of inefficient search in large instances, the proposed algorithm adaptively constraints the search range instead of relying on a completely random search, facilitating more effective exploration.
    In particular, the formal proof is provided that the introduction of plus transitions in the proposed algorithm ensures the connectivity of any node in the flip graph, representing a method of matrix multiplication.  
    Numerical experimental results demonstrate the effectiveness of the adaptive flip graph algorithm, showing a reduction in the number of multiplications for a $4\times 5$ matrix multiplied by a $5\times 5$ matrix from $76$ to $73$, and that from $95$ to $94$ for a $5 \times 5$ matrix multiplied by another $5\times 5$ matrix. These results are obtained in characteristic two.
\end{abstract}




\maketitle

\section{Introduction}
Efficient matrix multiplication is crucial for accelerating various computational processes, influencing overall computational speed significantly.  
Therefore, numerous studies have been undertaken to reduce the computational cost from various perspectives.
The groundbreaking work of Strassen demonstrated that two $2 \times 2$ matrices can be multiplied with only $7$ multiplications \cite{strassen1969gaussian}, leading to extensive studies on the complexity of matrix multiplication and the discovery of optimal multiplications for matrices beyond $2 \times 2$. 
For the complexity of matrix multiplication, recent developments continue to aim at finding the upper bound on the matrix multiplication exponent $\omega$. 
The current best-known bound is $\omega < 2.371552$ \cite{williams2023new}, improving the previous best bound of $\omega < 2.371866$ \cite{duan2022faster}. 
For optimal multiplication methods, notable progress has been made, reducing the multiplication of two $3 \times 3$ matrices to $23$ \cite{laderman1976noncommutative}, and the multiplication of two $4 \times 4$ matrices have been reduced to $49$ using recursively applying Strassen's algorithm \cite{strassen1969gaussian}.
Recently, a reinforcement learning approach has further refined the multiplication of two $4 \times 4$ matrices to $47$ \cite{Fawzi2022} and improved the best-known result for two $5 \times 5$ matrices from $98$ to $96$ multiplications.
Subsequently, focusing on the symmetry in matrix multiplications, a topology called flip graph has been introduced \cite{kauers2022flip}. 
By randomly repeating the local transition of vertices over the flip graph, the best-known result for two $5\times 5$ matrices has been improved from $96$ to $95$ multiplications \cite{kauers2022fbhhrbnrssshkalgorithm, kauers2022flip}.

In this study, we first identify two practical issues confronting the flip graph algorithm. 
The first issue is that the local transition can encounter troublesome ``states''  which limit further exploration and prevent updating the best number of multiplications. 
The second issue is that a random search on the flip graph becomes impractical as the matrix size increases within a reasonable computation time.  In fact, without setting the results obtained by the reinforcement learning to the initial condition, it would take too long to search from the initial condition and the best number of multiplications would not be reached \cite{kauers2022flip}. 
To address these issues, this study proposes an ``adaptive flip graph algorithm''. 
For the first issue, when encountering troublesome states, the algorithm adaptively transitions vertices in a direction that does not necessarily reduce the number of matrix multiplications, facilitating an escape from the states. 
We demonstrated that the connectivity of the flip graph is ensured by employing this transition. This proof confirms that all vertices in the flip graph are accessible, guaranteeing comprehensive exploration.
For the second issue,  instead of conducting a completely random search, 
the algorithm adaptively constraints the search range, enabling faster exploration.
Experiments show that the adaptive flip graph algorithm reduces the number of multiplication for a $4\times 5$ matrix and a $5\times 5$ matrix from $76$ to $73$ and the number of multiplication of two $5 \times 5$ matrices from $95$ to $94$, in characteristic 2. 

\section{Problem Formulation of Matrix Multiplication}\label{sec:problem-formulation}

Consider a field denoted as $K$, a $K$-Algebra denoted by $R$ and let $A \in R^{n \times m}$, $B \in R^{m \times p}$, respectively. The matrix multiplication $C=A B$ can be fully represented by the following tensor. 
\begin{definition}[Matrix Multiplication Tensor]
    For $n, m, p \in \mathbb{N}$, we define matrix multiplication tensor such that
    \begin{equation}
    \label{eq:tensorform}
        \mathcal{M}_{n, m, p} = \sum_{i,j,k=1}^{n, m, p} a_{i,j} \otimes b_{j, k} \otimes c_{k, i} \in K^{n \times m} \otimes K^{m\times p} \otimes K^{p \times n}.
    \end{equation}
    where $a_{u,v}$, $b_{u, v}$, $c_{u, v}$ refer to matrices of the respective format that have a $1$ at position $(u, v)$ and zeros in all other positions.    
\end{definition}

The number of multiplications to perform a matrix multiplication is connected to the rank of the matrix multiplication tensor. 
The rank of a tensor is defined as the smallest number $r$ such that $\mathcal{M}$ can be represented as a sum of $r$ rank-one tensors of the form $\mathcal{M} \in K^{n \times m} \otimes K^{m\times p} \otimes K^{p \times n}$.
We show the $2 \times 2$ matrix multiplication as a matrix multiplication tensor.

\begin{example}[Standard Classical Matrix Multiplication Algorithm for two $2\times 2$ matrices]
\label{ex:standard}
\begin{align*}
    \mathcal{M}_{2, 2,2} &= a_{1,1} \otimes b_{1, 1} \otimes c_{1, 1} 
    + a_{1,2} \otimes b_{2, 1} \otimes c_{1, 1} \\
    &+ a_{1,1} \otimes b_{1, 2} \otimes c_{2, 1} + a_{1,2} \otimes b_{2, 2} \otimes c_{2, 1} \\
    &+ a_{2,1} \otimes b_{1, 1} \otimes c_{1, 2} + a_{2,2} \otimes b_{2, 1} \otimes c_{1, 2} \\
    &+ a_{2,1} \otimes b_{1, 2} \otimes c_{2, 2} + a_{2,2} \otimes b_{2, 2} \otimes c_{2, 2}.
\end{align*}
In a similar way, for a matrix multiplication of an $n\times m$ matrix and an $m\times p$ matrix, $\mathcal{M}_{n, m, p}$  has a representation of the sum of $n\times m\times p$ rank-one tensors, which is referred to as the standard classical matrix multiplication algorithm. 
\end{example}

\begin{example}[Strassen's Algorithm]
\label{ex:Strassen}
The sum of $8$ linearly independent rank-one tensors need not necessarily have tensor rank $8$. Strassen shows that $\mathcal{M}_{2,2,2}$, defined as a sum of $8$ linearly independent rank-one tensors, can be represented as a following sum of $7$ linearly independent rank-one tensors:
    \begin{align*}
        \mathcal{M}_{2, 2,2} &= (a_{1, 1} + a_{2, 2}) \otimes (b_{1, 1} + b_{2,2}) \otimes (c_{1, 1} + c_{2,2}) \\
        &+ (a_{2, 1} + a_{2,2}) \otimes (b_{1, 1}) \otimes (c_{1, 2} - c_{2, 2}) \\
        &+(a_{1, 1}) \otimes (b_{1, 2} - b_{2, 2}) \otimes (c_{2, 1} + c_{2, 2}) \\
        &+(a_{2, 2}) \otimes (b_{2, 1}-b_{1, 1}) \otimes (c_{1, 1} + c_{1, 2}) \\
        &+(a_{1, 1} + a_{1, 2}) \otimes (b_{2, 2}) \otimes (c_{2, 1} - c_{1, 1}) \\
        &+ (a_{2, 1}-a_{1, 1}) \otimes (b_{1, 1} + b_{1, 2}) \otimes (c_{2, 2}) \\
        &+(a_{1, 2} - a_{2, 2}) \otimes (b_{2, 1} + b_{2, 2}) \otimes (c_{1, 1}).
    \end{align*}
\end{example}

In the subsequent discussion, we introduce the following definition for these rank-one tensors.
\begin{definition}[Matrix Multiplication Scheme]
A finite multiset $S$ of rank-one tensors, whose sum is $\mathcal{M}_{n, m, p}$, is referred to as an $(n, m, p)$-matrix multiplication scheme and $|S|$ is the rank of the scheme.
\end{definition}

\subsection{Flip Graph Algorithm}
A graph, known as the flip graph, has been introduced to represent the transformation relationships between matrix multiplication schemes. 
A search method on this graph has been proposed to reduce the rank of the scheme \cite{kauers2022flip}. This algorithm, which is the basis of the one employed in our study, is reviewed in this subsection. 

\begin{definition}[Flip]
    \label{def:flip}
    For two rank-one tensors satisfying $\alpha^{(i)} = \alpha^{(j)}$ in a scheme $S$, the following transformation is introduced: 
    \begin{align*}
        &\alpha^{(i)} \otimes \beta^{(i)} \otimes \gamma^{(i)} + \alpha^{(i)} \otimes \beta^{(j)} \otimes \gamma^{(j)}\\
        &\rightarrow \alpha^{(i)} \otimes (\beta^{(i)} + \beta^{(j)}) \otimes \gamma^{(i)} + \alpha^{(i)} \otimes \beta^{(j)} \otimes (\gamma^{(j)} - \gamma^{(i)}).
    \end{align*}
The scheme $S'$ obtained by this transformation is referred to as the \textit{Flip} of $S$.
This flip transformation also applies to any permutation of $\alpha$, $\beta$ and $\gamma$.
\end{definition}

\begin{definition}[Reduction]
\label{def:redcution}
For a non-empty set $I \subseteq \{1,2, \ldots, |S|\}$ in a scheme $S$, another scheme $S'$ with a rank\\ $|I| - \text{rank}\left(\sum_{i \in I}  \beta^{(r)} \otimes \gamma^{(r)}\right)$ smaller than the rank of S can be constructed using Gaussian elimination if the following conditions hold:
    \begin{enumerate}
        \item $\dim\langle\alpha^{(i)}\rangle_{i \in I} = 1$, 
        \item $\text{rank}\left(\sum_{i \in I}  \beta^{(r)} \otimes \gamma^{(r)}\right) < |I|$. 
    \end{enumerate}
The scheme $S'$ obtained by this transformation is referred to as the \textit{Reduction} of $S$.
This reduction transformation also applies to any permutation of $\alpha$, $\beta$ and $\gamma$.
\end{definition}

\begin{definition}[Flip Graph]
    Let $n, m, p \in \mathbb{N}$, and let $V$ be the set of all schemes for $(n, m, p)$-matrix multiplication. Define the following sets of edges: 
    \begin{enumerate}
        \item[] $E_1$:\{(S, S') $|$ $S'$ is a flip of $S$\},
        \item[] $E_2$:\{(S, S') $|$ $S'$ is a reduction $S$\}. 
    \end{enumerate}
The graph $G(V, E_1 \bigcup E_2)$ is then referred to as the $(n, m, p)$-flip graph. 
\end{definition}

The search method for a small rank scheme based on flip graphs is a random search, starting from an initial scheme, randomly choosing an edge in $E_1$ and seeking a vertex connected to $E_2$. The pseudo-codes of the conventional flip graph algorithm are shown in Algorithms~\ref{alg:random-search} and \ref{alg:flip-graph-algorithm}. 

\begin{algorithm}[H]
    \caption{Random Search}
    \label{alg:random-search}
    \begin{algorithmic}[1]
    \Function{random\_search}{$S$}
    \If{$E_{1} \neq \emptyset$}
        \State $S \gets$ Flip of $S$ by Definition \ref{def:flip}
    \Else
        \State \Return $S$
    \EndIf
    \If{$E_{2} \neq \emptyset$}
        \State $S \gets$ Reduction of $S$ by Definition \ref{def:redcution}
    \EndIf
    \State \Return $S$
    \EndFunction
    \end{algorithmic}
\end{algorithm}

\begin{algorithm}[H]
    \caption{Flip Graph Algorithm}
    \label{alg:flip-graph-algorithm}
    \begin{algorithmic}[1]
    \State \textbf{Input:} Initial scheme $S_{0}$ and total iteration $T$
    \State \textbf{Output:} $S_{T}$
    \For{$t = 1$ to $T$}
        \State $S_{t}\gets$ \Call{random\_search}{$S_{t-1}$} from Algorithm \ref{alg:random-search}.
    \EndFor
    \end{algorithmic}
\end{algorithm}

\section{Limitation and Practical Issues of Flip Graph Algorithm}\label{sec:issue-flip-algorithm}
In our exploration using a random search on the flip graph, we have encountered cases where finding efficient larger matrix multiplication schemes makes it challenging to find vertices on the flip graph that can reduce the rank of the scheme. These cases can be divided into two main types: 
\subsection{Non-Reduction States in Principle}
Numerical experiments have revealed a phenomenon where, for the $(4,4,4)$-matrix multiplication scheme, the rank of the scheme does not decrease during the search, even with increasing the number of iterations. This phenomenon is considered to be due to the absence of vertices that can reduce the rank among the vertices accessible by the flip algorithm. Such a state, called a ``non-reduction state", becomes more common as the matrix size increases.  

\subsection{Practical Non-Reducibility for Large Matrix Multiplication}
In random search experiments for the $(5,5,5)$-matrix multiplication scheme, we observed another phenomenon: the rank does not reduce at all from the initial scheme, which is the standard classical matrix multiplication algorithm as shown in Example~\ref{ex:standard}.
Notably, this phenomenon is not observed for a $4\times 4$ matrix or lower and only appears for order $5$ or larger matrices. This observation implies that as the size of the tensor increases, reaching reducible vertices by the random search becomes more challenging.  

The following argument shows that this issue is distinct from the non-reduction states described above. 
Suppose we take the standard classical matrix multiplication algorithm as in Example~\ref{ex:standard} for the initial scheme.
If the matrix is decomposed into submatrices, Strassen's algorithm as in Example~\ref{ex:Strassen} can eventually be applied to a $2\times 2$ submatrix, meaning that there is a reducible vertex in the neighborhood of the initial scheme on the flip graph. Thus, this consideration shows that this initial scheme is not in non-reduction states. 
Even then, it is challenging in a random search to start the rank reduction from the initial scheme. 
\section{Adaptive Flip Graph Algorithm}
Considering the limitations clarified, the next interesting study is the development of an adaptive flip strategy for efficient search on the flip graph for possible rank reduction. This adaptive approach could involve dynamically adjusting the search criteria based on the state during the search. Such an extension could improve the effectiveness of the flip graph algorithm and potentially reduce the challenges observed in our study. Here, we propose two adaptive approaches as solutions to each of the two limitations outlined above. 

\subsection{Plus Transition}
To address the former issue, we propose the plus transition, which introduces a transition that increases the rank of the scheme. This provides a strategy to escape from the non-reduction state described above, allowing the exploration of schemes that are inaccessible by the flip graph algorithm. The procedure for increasing the rank is defined as a "Plus" transformation.   

\begin{definition}[Plus]
    \label{def:plus}
    For a scheme $S$ the $\alpha^{(i)} \neq \alpha^{(j)}, \beta^{(i)} \neq \beta^{(j)}, \gamma^{(i)} \neq \gamma^{(j)}$ are satisfied, the ``Plus" transformation from $S$ to $S'$ is defined as follows: 
    \begin{align*}
        &\alpha^{(i)} \otimes \beta^{(i)} \otimes \gamma^{(i)} + \alpha^{(j)} \otimes \beta^{(j)} \otimes \gamma^{(j)}\\
        &\rightarrow \alpha^{(i)} \otimes \beta^{(i)} \otimes \gamma^{(i)} + \alpha^{(i)} \otimes \beta^{(j)} \otimes \gamma^{(j)}\\& + (\alpha^{(j)} - \alpha^{(i)}) \otimes \beta^{(j)} \otimes \gamma^{(j)}\\
        &\rightarrow \alpha^{(i)} \otimes (\beta^{(i)} + \beta^{(j)}) \otimes \gamma^{(i)} + \alpha^{(i)} \otimes \beta^{(j)} \otimes (\gamma^{(j)} - \gamma^{(i)})\\& + (\alpha^{(j)} - \alpha^{(i)}) \otimes \beta^{(j)} \otimes \gamma^{(j)}.
    \end{align*}
The scheme $S'$ obtained by this transformation is referred to as the \textit{Plus} of $S$.
This plus transformation also applies to any permutation of $\alpha$, $\beta$ and $\gamma$.
\end{definition}

The pseudo-code is given in Algorithm~\ref{alg2}. This algorithm, in addition to the random search from Algorithm \ref{alg:flip-graph-algorithm}, transitions from the scheme $S$ to its plus of $S$ if the rank does not decrease after a certain number of updates.

\begin{algorithm}[]
    \caption{Plus Transition}
    \label{alg2}
    \begin{algorithmic}[1]
    \State \textbf{Input:} Initial scheme $S_{0}$, total iteration $T$, and plus flag $L$
    \State \textbf{Output:} $S_{T}$
    \State $l \gets 0$
    \For{$t = 1$ to $T$}
        \State $S_{t}\gets$ \Call{random\_search}{$S_{t-1}$} from Algorithm \ref{alg:random-search}
        \If{$S_{t}$ is reduction of $S_{t-1}$}
            \State $l \gets 0$
        \Else
            \State $l \gets l+1$
        \EndIf
        \If{$l > L$}
            \State $S_{t} \gets$ Plus $S_{t}$ by Definition~\ref{def:plus}
            \State $l \gets 0$
        \EndIf
    \EndFor
    \end{algorithmic}
\end{algorithm}
\subsection{Edge Constraints}
To address the latter issue, we propose an alternative approach called edge constraints. This approach strategically reduces the size of the flip graph by restricting the edges to be transitioned in the flip graph. Generally, the initial scheme of a tensor $\mathcal{M}_{n,m,p}$ partially contains the initial scheme of the tensor $\mathcal{M}_{n,m,p-1}$. Thus, flip graphs consisting only of portions of initial schemes with small tensors such as $\mathcal{M}_{2,2,2}$ are explored first, and the search graph is gradually expanded by relaxing the constraints.  The pseudo-code for this algorithm is given in Algorithm~\ref{alg3}.

\begin{algorithm}[H]
    \caption{Edge Constraints}
    \label{alg3}
    \begin{algorithmic}[1]
    \State \textbf{Input:} Initial $(n,m,p)$-matrix multiplication scheme $S_{0}$, and total iterations $\mathcal{T}=\{T_{i}\}$
    \State \textbf{Output:} $S_{T}$
    \State Initialize $(n',m',p')\gets(2,2,2)$  
    \For{$i = 1$ to $n + m + p - 5$}
        \For{$t = 1$ to $T_{i}$}
            \State $S_{t}\gets$ \Call{random\_search}{$S_{t-1}$} in Algorithm \ref{alg:random-search}
        \EndFor
        \State $\displaystyle \min_{n^{\prime} < n, m^{\prime} < m, p^{\prime} < p} (n^{\prime}, m^{\prime}, p^{\prime}) \gets \min_{n^{\prime} < n, m^{\prime} < m, p^{\prime} < p} (n^{\prime}, m^{\prime}, p^{\prime})+1$
    \EndFor
    \end{algorithmic}
\end{algorithm}

\section{Connectivity of flip graphs through plus transitions}
In this section, we demonstrate that introducing plus transitions ensures the connectivity of the flip graph for $K = \mathbb{Z}2$, facilitating the exploitation of any matrix multiplication schemes.
\begin{lemma}
    \label{lemma:rank}
    In the $(n,m,p)$-matrix multiplication scheme S, it holds that $rank([\alpha^{(1)}, \alpha^{(2)}, ..., \alpha^{(R)}]) = n*m$. Similarly, for $\beta$ and $\gamma$, $rank([\beta^{(1)}, \beta^{(2)}, ..., \beta^{(R)}]) = m*p$ and $rank([\gamma^{(1)}, \gamma^{(2)}, ..., \gamma^{(R)}]) = p*n$.
\end{lemma}

\begin{proof}
    We show that $\text{rank}([\alpha^{(1)}, \alpha^{(2)}, \ldots, \alpha^{(R)}]) = n \times m$. The proof for $\beta$ and $\gamma$ follows a similar argument.\\
    By the definition of matrix multiplication,
    \begin{equation}
        \label{eq:definition of matrix multiplication}
        C_{i, j} = \sum_{k = 1}^m A_{i, k}B_{k, j}
    \end{equation}
    Considering the definition of the matrix multiplication tensor from Equation~\ref{eq:tensorform},
    \begin{align}
        C_{i, j} &= \sum_{r = 1}^R \gamma_{j, i}^{(r)} \left(\sum_{i_a, j_a = 1}^{n, m} \alpha_{i_a, j_a}^{(r)}A_{i_a, j_a}\right) \left(\sum_{i_b, j_b = 1}^{m, p} \beta_{i_b, j_b}^{(r)}B_{i_b, j_b}\right)\\
        \intertext{From Equation~\ref{eq:definition of matrix multiplication}, noting that the sum of all terms excluding $B_{k, j}$ is zero, we have}
        &= \sum_{r = 1}^R \gamma_{j, i}^{(r)} \left(\sum_{i_a, j_a = 1}^{n, m} \alpha_{i_a, j_a}^{(r)}A_{i_a, j_a}\right) \left(\sum_{k = 1}^{m} \beta_{k, j}^{(r)}B_{k, j}\right)\\
        &= \sum_{k = 1}^{m}\left(\sum_{r = 1}^R \beta_{k, j}^{(r)} \gamma_{j, i}^{(r)} \left(\sum_{i_a, j_a = 1}^{n, m} \alpha_{i_a, j_a}^{(r)}A_{i_a, j_a}\right)\right)B_{k, j}.
    \end{align}
    Comparing this with Equation~\ref{eq:definition of matrix multiplication}, we find
    \begin{align}
        A_{i, k} &= \sum_{r = 1}^R \beta_{k, j}^{(r)} \gamma_{j, i}^{(r)} \left(\sum_{i_a, j_a = 1}^{n, m} \alpha_{i_a, j_a}^{(r)}A_{i_a, j_a}\right),\\
        a_{i, k} &= \sum_{r = 1}^R \beta_{k, j}^{(r)} \gamma_{j, i}^{(r)} \alpha^{(r)}.
    \end{align}
    
    Thus, for any $i, k$, $a_{i, k}$ can be expressed as a linear combination of the vectors $\alpha^{(r)}$. This implies that the matrix constructed by arranging the vectors $\alpha$ for scheme S is of full rank, concluding that $\text{rank}([\alpha^{(1)}, \alpha^{(2)}, \ldots, \alpha^{(R)}]) = n \times m$.
\end{proof}

Regarding the scheme transformation by the plus operation as defined in Definition~\ref{def:plus}, the first operation 
\[\alpha \otimes \beta \otimes \gamma \rightarrow \alpha' \otimes \beta \otimes \gamma + (\alpha - \alpha') \otimes \beta \otimes \gamma\]
is interpreted as splitting the rank-1 tensor $\alpha \otimes \beta \otimes \gamma$ by $\alpha'$. In this transition, $\alpha'$ is limited to those $\alpha$ existing within the scheme $S$. 
Under this limitation, we prove that the flip graph is connected on $\mathbb{Z}2$ as follows.
\begin{theorem}
    The $(n, m, p)$-flip graph becomes connected over $\mathbb{Z}_2$ through the addition of edges by plus transformation.
\end{theorem}

\begin{proof}
    For any two schemes $S$ and $S'$,  We will show that there is a path from $S$ to $S'$. If $|S| < |S'|$, $S$ is transformed by applying number of $|S'| - |S|$ arbitrary plus transformations. Let $S_1$ be the resulting scheme. If $|S| \geq |S'|$, let $S_1 = S$.
    Next, the elements of $S_1$ and $S'$ are ordered, and for any given $k$, the $k$-th element of $S_1$ is transformed into the $k$-th element of $S'$. The elements added by the plus transformation are added after the $k$-th element. According to Lemma~\ref{lemma:rank}, the $\alpha$ component of the $k$-th element of $S'$ can be expressed as a linear combination of $\alpha$ in $S_1$. Therefore, by applying the plus transformation repeatedly, the $\alpha$ component of the $k$-th element in $S_1$ can be transformed into that of $S'$. This process is similarly applied to the $\beta$ and $\gamma$ components, ensuring that the $k$-th element of $S_1$ matches that of $S'$. Let $S_2$ denote this scheme.
    
    For the elements beyond the $|S'|$-th position in $S_2$, the following transformations are performed. According to Lemma~\ref{lemma:rank}, for any $i,j$, it is possible to represent $a_{i,j}$ as a linear combination of $\alpha$ in $S_2$. Therefore, by repeatedly applying plus transformations, a scheme $S_3$ can be constructed such that for any $i,j$, $a_{i,j} \in \{\alpha^{(r)}\}$.
    Next, using only those $\alpha^{(r)}$ that are in $\{a_{i,j}\}$ for splitting , plus transformations are performed to obtain a scheme $S_4$, where $\{\alpha^{(r)}|r > |S'|\} \subset \{a_{i,j}\}$.
    Applying the same operations from $S_2$ to $S_4$ to $\beta$ and $\gamma$, the scheme $S_5$ is obtained.
    
    Finally, in scheme $S_5$, for elements beyond the $|S'|$-th position, the reductions are repeatedly applied to any pairs of rank-1 tensors. The elements beyond the $|S'|$-th position in $S_5$ are expressed as $a_{i,j} \otimes b_{k,l} \otimes c_{m,n}$, and their sum is zero. Hence, this operation eliminates all such elements, resulting in $S_6 = S'$.
\end{proof}

In this proof, it is proven that the scheme $S$ can be transformed into $S'$ through a finite number of transformations. This leads to the following corollary, which appears almost trivial but proves the validity of the basic strategy of the adaptive flip-graph algorithm:  
\begin{corollary}
    A scheme with the minimal rank can be reached from the initial scheme by a finite number of transitions on the flip graph with edges of plus transformations. 
\end{corollary}

\section{Results}
This section aims to evaluate the effects of the plus transitions and edge constraints, respectively and presents the results of an adaptive flip graph algorithm that combines these methods for various matrix multiplications. 
Sec.~\ref{subsec:plus-experiment} demonstrates that the plus transition enables an efficient escape from non-reduction states and facilitates the search for better schemes.
Sec.~\ref{subsec:edge-constraint-experiment} shows that the edge constraints accelerate the exploration, in particular for escaping from the initial scheme. 
Finally, Sec.~\ref{subsec:results-adaptive-algorithm} presents the results of adaptive flip graph algorithms for various matrix multiplications.
In the following numerical experiments, we restricted matrix multiplication schemes to ground fields of characteristic two, $K = \mathbb{Z}_2$.


\subsection{Effectiveness of Plus Transition Escaping Non-Reduction States}\label{subsec:plus-experiment}
We demonstrate that the introduction of plus transitions indeed solves the issue of non-reduction states. 
For comparison, we applied the flip graph algorithm both with and without the plus transitions to the $(4,4,4)$-matrix multiplication scheme, using the standard classical matrix multiplication algorithm as the initial scheme. In this experiment, the total number of iterations $T$ was set to $10^{8}$, and the plus flag $L$ in Algorithm~\ref{alg2} was set to $5\times 10^{3}$. 
Figure \ref{fig:Plus_4} illustrates the dependence of the rank of the scheme on the number of iterations in the flip graph algorithm.
As shown in the figure, both algorithms begin to decrease in rank from the initial scheme after about $10^4$ iterations, and there is little difference up to $10^6$ iterations. Beyond that point, however, there is a significant difference. Without the plus transition, the rate of decrease in rank slows down, while with the plus transition, it continues to decrease steadily.  
These results suggest that employing the plus transition can effectively escape from non-reduction states.

\begin{figure}[]
    \centering
    \includegraphics[width=\linewidth, bb = 0 0 360 252]{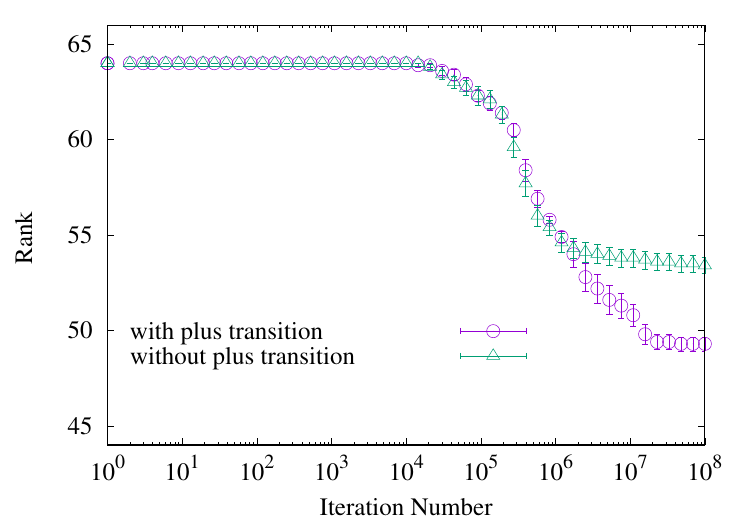}
    \caption{Dependence of the minimum rank found in the search on the number of iterations in the $(4,4,4)$-matrix multiplication scheme. Results for the flip graph algorithm with the plus transition are marked with circles and those without the plus transition are marked with triangles. Error bars represent the standard error obtained by 10 independent runs. }
    \label{fig:Plus_4}
\end{figure}

\subsection{Acceleration by Edge Constraints in Large Matrix Multiplications}\label{subsec:edge-constraint-experiment}

In this subsection, we address the issue of practical non-reducibility for large matrix multiplications, focusing in particular on the $(5,5,5)$-matrix multiplication scheme. We demonstrate that this issue is significantly diminished by the edge constraints introduced above. 
To conduct our comparative analysis, we applied the flip graph algorithm both with and without the edge constraints to the $(5,5,5)$-matrix multiplication scheme, using the standard classical matrix multiplication algorithm as the initial scheme. 
In this experiment, the total number of iterations was set to $\mathcal{T}=\{T_{i}\}$ with $\sum_{i} T_{i} = 10^8$ in Algorithm~\ref{alg3}, and the split points were equally spaced on a logarithmic scale. 
As shown in Figure \ref{fig:Edge_constraints}, when employing only the flip graph algorithm, even $10^8$ iterations failed to find schemes with ranks lower than 125, indicating the seriousness of the issue of practical non-reducibility. However, the introduction of the edge constraints facilitated the identification of reducible vertices in about $10^2$ iterations and began to decrease the rank. 
This result shows that edge constraints are an effective method for addressing the issue of practical non-reducibility.

\begin{figure}[]
    \centering
    \includegraphics[width=\linewidth, bb = 0 0 360 252]{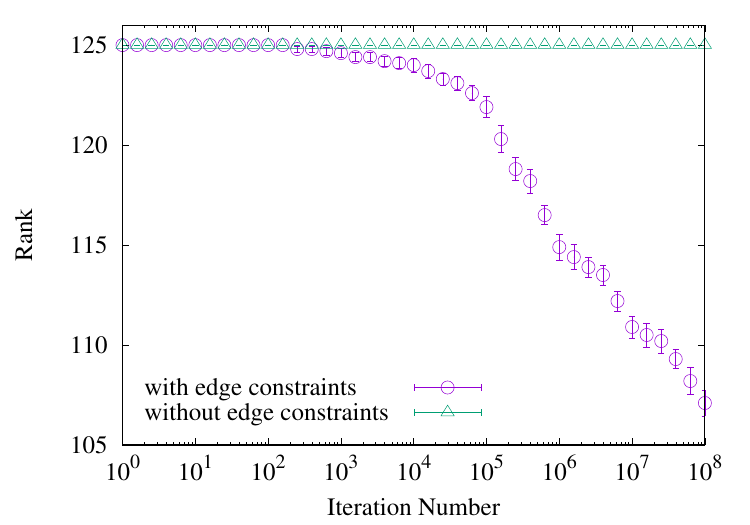}
    \caption{Dependence of the minimum rank found in the search on the number of iterations in the $(5,5,5)$-matrix multiplication scheme. Results for the flip graph algorithm with the edge constraints are marked with circles and those without the edge constraints are marked with triangles. Error bars represent the standard error obtained by 10 independent runs. }
    \label{fig:Edge_constraints}
\end{figure}

\subsection{Results of Adaptive Flip Graph Algorithm for Various Matrix Multiplications}\label{subsec:results-adaptive-algorithm}
We have confirmed that both plus transitions and edge constraints improve the efficiency of the flip graph algorithm from different perspectives. The adaptive flip graph algorithm, incorporating these two ideas, was applied to the $(n,m,p)$-matrix multiplications in characteristic 2. 
For the results of the adaptive flip graph algorithm, the plus flag was set to $L = 5\times 10^{3}$ for all edge constraints. In the case of the maximal case of the $(5,5,5)$-matrix multiplication, the total number of iterations of edge constraints was set as $\mathcal{T}=\{5\times 10^7$, $5\times 10^7$, $5\times 10^7$, $5\times 10^7$, $10^8$, $1.5\times 10^8$, $2.5\times 10^8$, $7.5\times 10^8$, $1.5\times 10^9$, $3\times 10^9\}$ and the edge constraints were repeated $10^2$ times, starting from the previous edge constraints scheme with lowest scheme rank. 
As shown in Table \ref{table:result-matrix}, the adaptive flip graph algorithm achieved updated minimum ranks for the $(4,5,5)$ and $(5,5,5)$-matrix multiplication schemes, while reproducing the previously known minimum ranks for other smaller matrix multiplication schemes. 
The schemes we found for $(4,5,5)$ and $(5,5,5)$-matrix multiplications are explicitly shown in \url{https://github.com/Yamato-Arai/adap}.

\begin{table}[]
    \caption{Comparison of the best-known rank and the rank found by the adaptive flip graph algorithm in characteristic 2 for various matrix multiplications, denoted as ``Adap. Flip''. 
    The table includes results obtained by the Alpha Tensor-based reinforcement learning \cite{Fawzi2022} denoted as ``AT'', those obtained by the flip graph algorithm starting from the standard classical matrix multiplication algorithms \cite{kauers2022flip} denoted as ``Flip'' and those obtained by the flip graph algorithm starting from the ``AT'' results examined only in the $(5,5,5)$-matrix multiplication denoted as ``FwAT'' \cite{kauers2022flip}. 
      }
  \centering
  \begin{tabular}{c|lc|ccc|c}
    \toprule
    $(n,m,p)$ & Best Method &Best & AT & Flip & FwAT & Adap. Flip\\
    \midrule
    $(2,2,2)$ & \cite{strassen1969gaussian} & 7  & 7 & 7 & -- & 7 \\
    $(2,2,3)$ & \cite{hopcroft1971minimizing} & 11 & 11 & 11 & -- & 11 \\
    $(2,2,4)$ & \cite{hopcroft1971minimizing} & 14 & 14 & 14 & -- & 14 \\
    $(2,3,3)$ & \cite{hopcroft1971minimizing} & 15 & 15 & 15 & -- & 15 \\
    $(2,2,5)$ & \cite{hopcroft1971minimizing} & 18 & 18 & 18 & -- & 18 \\
    $(2,3,4)$ & \cite{hopcroft1971minimizing} & 20 & 20 & 20 & -- & 20\\
    $(3,3,3)$ & \cite{laderman1976noncommutative} & 23 & 23 & 23 & -- & 23\\
    $(2,3,5)$ & \cite{hopcroft1971minimizing} & 25 & 25 & 25 & -- & 25\\
    $(2,4,4)$ & \cite{hopcroft1971minimizing} & 26 & 26 & 26 & -- & 26\\
    $(3,3,4)$ & \cite{smirnov2013bilinear} & 29 & 29 & 29 & -- & 29\\
    $(2,4,5)$ & \cite{hopcroft1971minimizing} & 33 & 33 & 33 & -- & 33\\
    $(3,3,5)$ & \cite{smirnov2013bilinear} & 36 & 36 & 36 & -- & 36\\
    $(3,4,4)$ & \cite{smirnov2013bilinear} & 38 & 38 & 38 & -- & 38\\
    $(2,5,5)$ & \cite{hopcroft1971minimizing} & 40 & 40 & 40 & -- & 40\\
    $(3,4,5)$ & \cite{Fawzi2022} & 47 & 47 & 47 & -- & 47\\
    $(4,4,4)$ & \cite{Fawzi2022} & 47 & 47 & 47 & -- & 47\\
    $(3,5,5)$ & \cite{sedoglavic2021tensor} & 58 & 58 & 58 & -- & 58\\
    $(4,4,5)$ & \cite{kauers2022flip} & 60 & 63 & 60 & -- & 60\\
    $(4,5,5)$ & \cite{Fawzi2022} & 76 & 76 & 76 & -- & $\bm{73}$ \\
    $(5,5,5)$ & \cite{kauers2022flip} & 95 & 96 & 97 & 95 & $\bm{94}$ \\
    \bottomrule
  \end{tabular}
    \label{table:result-matrix}
\end{table}


\section{Discussion and Summary}
\paragraph{Evaluation and Implication of Results}
We compare our study with previous studies from the perspective of computational resources.
All numerical experiments in this study were performed on an Intel(R) Core(TM) i7-8650U CPU \@1.90GHz 2.11 GHz, single CPU thread.
While the previous reinforcement learning study using the Alpha Tensor \cite{Fawzi2022} employed a 64-core TPU and the original flip graph algorithm \cite{kauers2022flip} required massive parallel computational resources, our method demonstrated high efficiency with only a single CPU thread.

Moreover, we contrast our approach with the method presented in the original flip graph algorithm \cite{kauers2022flip} with respect to the initial schemes. The previous study found the scheme with rank 95 for the $(5,5,5)$-matrix multiplication, using the Alpha Tensor results as the initial scheme. In contrast, when using the standard classical matrix multiplication algorithm as the initial scheme, the best scheme found was the scheme of rank 97. It should be emphasized that our method successfully found the scheme of rank 94 using the standard classical matrix multiplication algorithm as the initial scheme.
\paragraph{Future Work and Conclusion}
This study identified and addressed issues in a random search on the flip graph by combining the plus transitions and edge constraints, resulting in significant improvements in the $(4,5,5)$ and $(5,5,5)$-matrix multiplications. Our future objective is to develop strategies for handling matrices larger than sixth order, using both existing and innovative methods, with the ultimate objective of finding a matrix multiplication algorithm that substantially outperforms Strassen's algorithm.

\begin{acks}
We would like to gratefully acknowledge Jun Takahashi and Yasushi Nagano for their valuable insights and discussions at the outset of this study. 
This work was supported by JST Grant Number JPMJPF2221 and JSPS KAKENHI Grant Number 23H01095.
\end{acks}

\printbibliography





\end{document}